\documentclass[a4paper,12pt]{article}

\usepackage{amsfonts,amssymb,amstext,amsbsy,amsopn,amsthm}

\newtheorem{thm}{Theorem}[section]
\newtheorem{lem}[thm]{Lemma}
\newtheorem{cor}[thm]{Corollary}
\newtheorem{prop}[thm]{Proposition}

\newtheorem{defn}[thm]{Definition}
\newtheorem{assu}[thm]{Assumption}

\newtheorem{rem}[thm]{Remark}

\newcommand{\field}[1]{\mathbb{#1}}
\newcommand{\N}{\field{N}}
\newcommand{\R}{\field{R}}

\newcommand{\be}[1]{\begin{equation}\label{e-#1}}  
\newcommand{\ee}{\end{equation}}  

\newcommand{\fx}{x}

\newcommand{\hone}{W^{1,2}}

\newcommand{\lei}{{L^1\ab}}
\newcommand{\lzw}{{L^2\ab}}
\newcommand{\lun}{{L^\infty\ab}}

\newcommand{\A}{{0}}
\newcommand{\B}{{1}}

\newcommand{\ab}{{(\A,\B)}}

\newcommand{\vr}{{\bf r}}

\newcommand{\assuref}[1]{Assumption~\ref{#1}}
\newcommand{\thmref}[1]{Theorem~\ref{#1}}
\newcommand{\corref}[1]{Corollary~\ref{#1}}
\newcommand{\remref}[1]{Remark~\ref{#1}}
\newcommand{\proref}[1]{Proposition~\ref{#1}}

\newcommand{\lemref}[1]{Lemma~\ref{#1}}

\newcommand{\spec}{spec}

\newcommand{\tr}{{\rm tr}}

\DeclareMathOperator*{\essinf}{ess\,inf}
\DeclareMathOperator*{\esssup}{ess\,sup}

\title{\vspace{-2.0cm}A Kohn-Sham system at zero temperature}

\author{
Horia Cornean\\
Department of Mathematical Sciences\\
Aalborg University\\
Fredrik Bajers Vej 7E\\
DK-9220 Aalborg, Denmark\footnotemark[1]
\and
Kurt Hoke\footnotemark[2]\\
Hagen Neidhardt\footnotemark[2]\\
Paul Nicolae Racec\footnotemark[2]\\ 
Joachim Rehberg\footnotemark[2]\\
Weierstrass Institute  for\\
Applied Analysis and Stochastics\\
 Mohrenstr. 39\\
D-10117 Berlin, Germany
}

\date{January 17, 2008}

\begin{document}

\maketitle

\begin{abstract}
\noindent
An one-dimensional Kohn-Sham system for spin particles 
is considered which effectively describes semiconductor 
{nano}structures and which is investigated at zero temperature.
We prove the existence of solutions and derive a priori
estimates. For this purpose we find estimates for eigenvalues of the Schr\"odinger
operator with effective Kohn-Sham potential and obtain $W^{1,2}$-bounds of
the associated particle density operator.
Afterwards, compactness and continuity results allow to apply Schauder's
fixed point theorem. In case of vanishing exchange-correlation potential uniqueness is
shown by monotonicity arguments. Finally, we investigate the behavior of the system 
if the temperature approaches zero.
\end{abstract}



\footnotetext[1]{cornean@math.aau.dk}
\footnotetext[2]{hoke@wias-berlin.de, neidhard@wias-berlin.de,
  racec@wias-berlin.de, \\rehberg@wias-berlin.de}

\section{Introduction}

Hohenberg and Kohn have shown in \cite{HK1} that the ground state of an
$N$-body quantum system at zero temperature 
is completely determined by the particle density. Nowadays that paper is
considered as the starting point of the so-called density functional
theory. The main  advantage of this approach is that the
description of an $N$-body quantum problem can be reduced to an effective one-body
system. A shortcoming of \cite{HK1} is that only the existence of such a functional,
 depending on the  particle density, was shown, but the functional was not given 
explicitely. In \cite{KS1} 
Kohn and Sham have indicated such functionals which {$N$-electrons} look like
\begin{eqnarray*}
\lefteqn{
\mathcal E[u] = -\frac{\hbar^2}{2m}
                 \sum^N_{n=1}\int |\nabla\varphi_n(\vr)|^2 d^3\vr}\\
& &\hspace{-1.0cm}
-{\frac{q^2}{4 \pi \epsilon_0}} 
 \sum^M_{k=1}\int \frac{Z_k u(\vr)}{|\vr - {\bf R}_k|}d^3\vr +
\frac{q^2}{2}\int\int\frac{u(\vr)u(\vr')}
{{4 \pi \epsilon_0}|\vr-\vr'|}d^3\vr d^3\vr' + \int \varepsilon_{xc}[u](\vr)d^3\vr,
\end{eqnarray*}
where $M$ is the number of positive ions, 
$Z_k$ their {atomic} number, ${\bf R}_k$ the
positions of ions, $q$ is the magnitude of the elementary charge,
{$\epsilon_0$ is the vacuum permittivity} 
 and $\varepsilon_{xc}[u]$ is the so-called
exchange correlation energy density. 
The {particle} density $u$ is given by the expression
\begin{equation}\label{0.1}
u(\vr) := {2} \sum^N_{n=1}|\psi_n(\vr)|^2,
\end{equation}
where {$2$ counts for the spin degeneracy of the particles}
and $\psi_n$ are eigenfunctions satisfying the Kohn-Sham equation
\begin{displaymath}\hspace{-2.0cm}
\left(-\frac{\hbar^2}{2m}\Delta - \frac{q^2}{4 \pi \epsilon_0}
\sum^M_{k=1}\frac{Z_k}{|\vr - {\bf R}_k|} + 
q^2\int\frac{u(\vr')}{4 \pi \epsilon_0|\vr-\vr'|}d^3\vr' + V_{xc}[u](\vr)\right)\psi_n = E_n\psi_n.
\end{displaymath}
By $V_{xc}[u]$ the so-called exchange correlation potential is denoted, which is given by
\begin{displaymath}
V_{xc}[u](\vr):=  \frac{\partial (\varepsilon_{xc}[u](\vr))}{\partial u}.
\end{displaymath}
The potential
\begin{displaymath}
V_0(\vr) := - {\frac{q^2}{4 \pi \epsilon_0}}\sum^M_{k=1}\frac{Z_k}{|\vr - {\bf R}_k|},
\end{displaymath}
which is determined by the positive ions, can be regarded as a given external
potential. The potential
\begin{displaymath}
\varphi(\vr): = -q\int\frac{u(\vr')}{{4 \pi \epsilon_0}|\vr-\vr'|}d^3\vr' 
\end{displaymath}
is nothing else as the solution of the Poisson equation
\begin{equation}\label{0.0}
\Delta \varphi(\vr) = \frac{q u(\vr)}{{\epsilon_0}}.
\end{equation}
So we end up with a Schr\"odinger operator of the form
\begin{equation}\label{0.2}
H_V = -\frac{\hbar^2}{2m}\Delta + V
\end{equation}
with the effective Kohn-Sham potential
\begin{equation}\label{0.3}
V := V_0 + V_{xc}[u] - q\varphi,
\end{equation}
where $\varphi$ obeys the Poisson equation (\ref{0.0}). The model is
very flexible and widely applicable because the exchange correlation term
can be well adapted to a great variety of problems. This is one of 
the reasons why the approach was very successful in the last forty years,
and Kohn was awarded the Nobel prize in 1998 for that idea. 

We note that the Kohn-Sham system is quite similar to the so-called Hartree-Fock
approximation in $N$-body quantum systems \cite{LS1,L1}. 
The main difference is that the exchange correlation term for the Hartree-Fock system
is not local. However, performing the low density limit one obtains
the Hartree-Fock-Slater approximation \cite{Sl1,D1,BLS1}.
In this case the exchange correlation potential is  of the form
$V_{xc}[u](x) = -C_\alpha|u(x)|^\alpha$, $\alpha = 1/3$ in three dimensions (3D). For $\alpha =
2/3$ one gets another interesting approximation which is called the 
Thomas-Fermi correction. Usually, models of that type are summarized
as Schr\"odinger-Poisson-$X^\alpha$ systems, see \cite{BLSS1}, which
fit, of course, into the class of Kohn-Sham systems. In the following 
we do not restrict ourself to Schr\"odinger-Poisson-$X^\alpha$ systems, 
but consider a {larger} class of local and non-local exchange correlation terms
including the $X^\alpha$-models. Generally, we assume that the exchange
correlation is a non-linear mapping acting from the set of densities into the set {of} 
potentials obeying a certain continuity condition.

Note that all these considerations are made at zero temperature. An extension of the
 Hohenberg-Kohn approach to temperatures above zero was proposed by Mermin in \cite{Mer1}. 
He showed that the expression for the particle density then modifies to
\begin{equation}\label{0.4}
u(\vr) := {2}\sum^\infty_{n=1}\frac{1}{1 + e^{\beta(E_n - \mu)}}|\psi_n(\vr)|^2,
\end{equation}
such that
\begin{equation}\label{0.5}
N = \int u(\vr) d^3\vr =  {2} \sum^\infty_{n=1}\frac{1}{1 + e^{\beta(E_n -\mu)}},
\end{equation}
where $\beta := 1/kT$ and $\mu$ is the so called chemical potential.

There are many papers on the numerics of the Kohn-Sham system, but very few
on its mathematical analysis. In case of nonzero temperature and
bounded domains the system was analyzed in
\cite{KR1} and \cite{KR2}, where existence and a priori estimates were
shown. In \cite{PN1} the Schr\"odinger-Poisson-Slater system was investigated for a periodic
external potential $V_0$. The time dynamics of the
Schr\"odinger-Poisson-Slater system was considered in \cite{SS1}.

In the following we are going to investigate the zero and non-zero
temperature Kohn-Sham system with a general exchange correlation potential
for a planar semiconductor {nano}structure. The system reduces essentially to an effective one-dimensional
system. Since for one-dimensional systems the eigenvalues are simple, one avoids in this
way the occupation problem for the last eigenvalue at zero temperature,
if it is degenerated. We show the existence of solutions for such systems at 
non-zero and zero temperature. 
In particular, we prove that the solution is unique, if the exchange correlation potential is
 absent. In the zero temperature case this proof is based on an extension of the monotonicity 
for the negative particle density operator to non-smooth distribution functions like
Fermi-Dirac distribution function at zero temperature, see \cite{KNR1}.
Finally, we prove that the non-zero temperature solutions of the
Kohn-Sham system converge to those for zero temperature as the temperature goes to zero. 

The outline of the paper is as follows:
In Section 2 we derive an expression of the effective one-dimensional
particle density for a planar semiconductor {nano}structure. 
In Section 3 we introduce the mathematical setup of an one-dimensional 
Schr\"o\-dinger-Poisson system and make it
mathematically rigorous. Section 4 is devoted to the existence of
solutions. In Section 5 we prove the uniqueness of solutions, if the correction term is absent. 
Finally, in Section 6 we show the convergence of non-zero temperature
solutions to zero-temperature ones as the temperature goes to zero.

{\bf Notation:} In this paper the system is considered on the domain
$\Omega :=]0,1[$. For this reason we omit for all functional spaces the explicit indication
of this interval; e.g. write $L^1$ instead of $L^1(]0,1[)$ and so on.
We set $W^{1,2}_0 := \{f \in W^{1,2}: f(0)
  = f(1) = 0\}$. The space of antilinear forms on $W^{1,2}_0$ is denoted by $W^{-1,2}$. 
For Banach spaces $X$ and $Z$, we denote by $\mathcal B(X;Z)$ the space
of all linear, continuous operators from $X$ into $Z$. If $X=Z$ we write
$\mathcal B(X)$. Because of the numerous use of $X=L^2$ we introduce the abbreviation 
$\|\cdot\|=\|\cdot\|_{\mathcal{B}(L^2)}$.\\

\section{Particle density for planar nanostructures}

We consider a planar semiconductor nanostructure; that is, there is a sequence of layers
of different materials along the $x$-direction (i.e. a sequence of quantum wells and barriers)
 embedded between two thick layers of isolator. Then the wave functions of a particle 
(electron or hole) are given by
\begin{equation}\label{e-Psi_planar}
\Psi_{{\textbf k_\perp},l}({\textbf r})=\frac{ e^{i {\textbf k_\perp}{\textbf r_\perp}}}{2\pi} 
\psi_l(x),
\end{equation}
and the total energy of the particle is
\begin{equation}
E=\frac{\hbar^2 k_\perp^2}{2m_\perp}+\lambda_l,
\label{tot_en}
\end{equation}
where ${\textbf r_\perp}=(y,z)$ represents the transversal coordinates, 
${\textbf k_\perp}=(k_y,k_z)$ the transversal wave number and 
$m_\perp$ the effective mass in the transversal direction.
The wave functions along the $x$-direction and their corresponding energies
correspond to the one-dimensional  Schr\"odinger operator in the effective
mass-approximation (Ben-Daniel-Duke form)
\begin{equation}\label{e-schr0}
H_V :=  - \frac{\hbar^2}{2} \frac {d}{dx} \left(m^{-1} \frac {d}{dx} \right) + V,
\end{equation}
where $m = m(x)$ is the position dependent effective mass and 
$V$ is an effective Kohn-Sham potential to be specified later on. The embedding isolator layers
 impose homogeneous Dirichlet boundary conditions for the wave function along the $x$-direction,
providing a discrete spectrum of energies, $\lambda_l$, and defining in such 
a way a {\it closed} system in the $x$-direction.

Quantum mechanically, the particle density is given by a sum over all
states of their localization probability multiplied by their occupation
probability. 
For fermions (electrons and holes) the occupation probability
is given by the Fermi-Dirac function. 

At zero temperature, all states up to the Fermi energy $E_F$ 
are equally occupied, with probability $1$
and above $E_F$ all states are empty, i.e. occupation probability $0$. 
Thus, for zero temperature, the {particle} density is 
calculated quantum mechanically as
\begin{displaymath}
u({\textbf r})= 
{2} \underbrace{\int d{\textbf k_\perp} 
\sum_l}_{E \le E_F} |\Psi_{{\textbf k_\perp},l}({\textbf r})|^2, 
\end{displaymath}
where {$2$ counts for the spin degeneracy of the particles}.
Using the expression (\ref{e-Psi_planar}) of the wave functions for 
planar structures, one has only an $x-$ dependent particle density
\begin{displaymath}
u(x)= {2} \sum_{l=1}^{N_F} \frac{|\psi_l(x)|^2}{(2\pi)^2} 
          \int_{0}^{k_{\perp,F}^{(l)}} 1\; d{\textbf k_\perp} ,
\end{displaymath}
where the sum runs up to the last occupied level, 
i.e. $\lambda_{N_F} \le E_F$ and the integral is taken up to a maximum
value of the transversal wave number, 
$k_{\perp,F}^{(l)}=\sqrt{\frac{2m_\perp}{\hbar^2}\left(E_F-\lambda_l\right)}$ 
-- depending  on $l$. 
The integral over $d{\textbf k_\perp}$ can be performed, and one
obtains the {particle} density at zero temperature
for an effective one-dimensional system \cite{pr00} as
\begin{equation}\label{rho_0}
u(x)={2} \frac{m_\perp}{2\pi \hbar^2} 
      \sum_{l=1}^{N_F} |\psi_l(x)|^2 \left(E_F-\lambda_l\right).
\end{equation}

At temperature $T$ different from zero the particle density is given by
\begin{eqnarray*}
u({\textbf r}) &=&
{2} \underbrace{\int d{\textbf k_\perp} \sum_l}_{0 \le E \le \infty} 
|\Psi_{{\textbf k_\perp},l}({\textbf r})|^2 f_{FD}(E,\mu) \nonumber \\
&=&{2} \int_0^\infty d{\textbf k_\perp} \sum_{l=1}^\infty 
    |\Psi_{{\textbf k_\perp},l}({\textbf r})|^2 f_{FD}(E,\mu),
\end{eqnarray*}
where $f_{FD}(E,\mu)$ is the Fermi-Dirac distribution function
\begin{displaymath}
f_{FD}(E,\mu)=\frac{1}{1+e^{\frac{E-\mu}{k T}}},
\end{displaymath}
$k$ is the Boltzmann constant and $\mu$ is the chemical potential. Inserting 
(\ref{e-Psi_planar}) and using (\ref{tot_en}) one obtains an $x-$ dependent 
particle density
\begin{displaymath}
u(x)= {2} \sum_{l=1}^{\infty} \frac{|\psi_l(x)|^2}{(2\pi)^2}
       \int_{0}^{\infty} 
         f_{FD}\left(\frac{\hbar^2 k_\perp^2}{2m_\perp}+\lambda_l,\mu\right)
       d{\textbf k_\perp}. 
\end{displaymath}
This corresponds to the general form (\ref{e-cars}) used later on.
Also in this case one can perform the integral over $d{\textbf k_\perp}$, 
obtaining \cite{cahay87,zimmermann88}
\begin{equation} \label{e-einsetz}
u(x)={2} \frac{m_\perp}{2\pi \hbar^2}
      kT \sum_{l=1}^{\infty} |\psi_l(x)|^2 
      \ln\left(1+e^{\frac{\mu-\lambda_l}{kT}}\right),
\label{rho_T}
\end{equation}
which provides the expression for $f_\beta(x)$ used in Lemma \ref{l-monotonT} below.
Carrying out the limit $T \rightarrow 0$ in (\ref{e-einsetz}) one obtains the same expression
 for the particle density as in the zero temperature limit (\ref{rho_0}), because 
\begin{displaymath}
\lim_{T\rightarrow 0} kT \ln\left(1+e^{\frac{\mu-\lambda_l}{kT}}\right) =
\left( \mu - \lambda_l \right) \Theta(\mu-\lambda_l) 
\end{displaymath}
($\Theta$ being the Heaviside function) and it is known that the chemical
 potential for zero temperature equals the Fermi energy, 
$\lim\limits_{T\rightarrow 0} \mu = E_F$.

\section{Kohn-Sham systems}
\label{s-spsys}

\subsection{ Kohn-Sham systems in one dimension}

The Kohn--Sham system is a system of equations governing the
electrostatic potential $\varphi$ and the  particle density $u$ under
consideration.  
{Let us consider a system of electrons.}
The electrostatic potential and the particle density have to obey
Poisson's equation
\begin{eqnarray}\label{e-poi}
  -\frac {d}{dx} \bigl( \varepsilon \frac {d}{dx}\varphi \bigr)
  = D - q u                              
\end{eqnarray}
in the device domain $\Omega = (0,1)$ where
$q$ is the magnitude of the elementary charge, and
$\varepsilon=\varepsilon(\fx)$ denotes                                 
the dielectric permittivity. 
The right--hand side of (\ref{e-poi}) is a charge
distribution $D$ of ionized dopants
and the particle density $u$ which is defined below, see (\ref{e-cars}).
One has to supplement the Poisson equation (\ref{e-poi}) by boundary
conditions. Usually one chooses inhomogeneous Dirichlet boundary conditions
\begin{equation}  \label{e-bv}
\varphi (0)  = \varphi_0 \in \R
\quad \mbox{and} \quad 
\varphi(1) = \varphi_1 \in \R,
\end{equation}
 which model Ohmic contacts. 
A straightforward {calculation} shows that inhomogeneous boundary
conditions can be transformed into homogenous boundary
conditions. Indeed, introducing the function $\widetilde{\varphi}: [0,1] \longmapsto \R$,
\begin{displaymath}
\widetilde{\varphi}  := \varphi_0 + 
\frac{\varphi_1 - \varphi_0}{\int^1_0 \frac{1}{\varepsilon(t)}dt}
\int^x_0 \frac{1}{\varepsilon(t)}dt
\end{displaymath}
and setting $\phi := \varphi(x) - \widetilde{\varphi}$, $x \in
[0,1]$, one gets that $\phi$ satisfies the Poisson equation
\begin{equation}\label{e-poi1}
- \frac{d}{dx}\varepsilon\frac{d}{dx}\phi  = D - qu
\end{equation}
and obeys the {homogeneous} Dirichlet boundary conditions
\begin{equation}\label{e-bv1}
\phi(0) = 0 
\quad \mbox{and} \quad 
\phi(1) = 0.
\end{equation}
This gives rise to the Poisson operator 
${\mathcal P} :=  -\frac {d}{dx} \varepsilon \frac {d}{dx}$,
supplemented by homogeneous Dirichlet boundary conditions.

The particle density $u$ is computed by
the quantum mechanical expression
\begin{equation}\label{e-cars}
u(V)(\fx)
={2}\sum_{l=1}^{\infty}f(\lambda_l(V)-\mu_{f}(V))\,|\psi_l(V)(x)|^2, 
\quad \fx \in \ab,
\end{equation}
where {$2$ counts for the spin degeneracy},
{$f$ is a distribution function (to be specified later on),}
$\lambda_l =\lambda_l(V)$ are the eigenvalues and
$\psi_l=\psi_l(V)$ are the corresponding $L^2$-normalized orthogonal eigenfunctions
of the Schr\"odinger operator $H_V$, cf. (\ref{e-schr0}).
The chemical
potential $\mathcal \mu_f(V)$ is determined by the condition
\begin{displaymath}
{2}\sum_{l=1}^\infty f(\lambda_l(V)- \mu_f(V))=N\,,
\end{displaymath}
where $N$ is the number of particles ({electrons})-- which is considered within this paper as given and
 fixed, see (\ref{0.5}). The effective Kohn--Sham potential $V$  depends on the particle densities
and splits up in the following way
\begin{displaymath}
V=V(u) = {\Delta E} + V_{xc} (u ) - q \varphi ,
\end{displaymath}
where $\varphi$ denotes the electrostatic potential.
%
%
{The given external potential ${\Delta E}$ 
represents the band--edge offsets
of the nanostructure materials.}
$V_{xc}$ is the exchange--correlation potential,
which depends on the particle density, see Section 1. 

It is a widely discussed question
how to supplement the Schr\"odinger operators (\ref{e-schr0}) by
suitable boundary conditions (see \cite{Fr2,kerkhoven94,kerkhoven96}).
We choose {homogeneous} Dirichlet boundary conditions
\begin{displaymath}
\psi(0)=0 
\quad \mbox{and} \quad 
\psi(1)=0
\end{displaymath}
for all $\psi$ in the domain of the Schr\"odinger operator
$H_V$. {They are assisted by the fact that} 
if we assume a device structure which confines the particles
(i.e. closed system), 
then the particle densities vanish on the boundary
of $\Omega$.

\subsection{Rigorous mathematical formulation of the problem}\label{s-problem}

In this section we give a mathematical formulation of the Kohn-Sham system; in particular,
 we make precise in which spaces the corresponding operators act and the solutions 
are chosen. In view of typical applications \cite{gajewski93}, our mathematical model must
 necessarily cover semiconductor heterostructures, i.e.,  the coefficients of the Schr\"odinger
and the Poisson operator are in general discontinuous. This forecloses that
the domain of the Schr\"odinger operator is not lying in $W^{2,2}$ what is natural elsewhere,
see e.g. \cite{nier91,nier93}. Fortunately, in the one--dimensional case the $\hone$--calculus
already leads to satisfactory results. Let us introduce the Poisson operator. 
\begin{assu}\label{assu-diel}
{\rm
The dielectric permittivity $\varepsilon$ is a real, non-negative function obeying 
$\varepsilon \in L^\infty$ and $\frac{1}{\varepsilon} \in L^\infty$.
}
\end{assu}
\begin{defn}
{\rm 
Let Assumption \ref{assu-diel} be satisfied. We define the Poisson operator ${\mathcal P}:
 W^{1,2}_0 \longmapsto W^{-1,2}$ by
\begin{equation} \label{e-pois}
<{\mathcal P}v,w> = \int^1_0 \varepsilon(x)v'(x)\overline{w'(x)} dx, \quad u,w \in  W^{1,2}_0, 
\end{equation}
where here and in the sequel $<\cdot,\cdot>$ denotes the dual pairing between $W^{1,2}_0$ and 
$W^{-1,2}$.
}
\end{defn}
One easily estimates
\begin{equation}
\left|<{\mathcal P}v,w>\right| \le
\|\varepsilon\|_{L^\infty}\|v\|_{W^{1,2}}\|w\|_{W^{1,2}}
\end{equation}
for $u,w \in W^{1,2}_0$.
Consequently, ${\mathcal P}:W^{1,2}_0 \mapsto W^{-1,2}$ is well defined and continuous. 
Furthermore, we have
\begin{displaymath}
\|\phi\|^2_{W^{1,2}_0} \le 
\|2/\varepsilon\|_{L^\infty} \int^1_0 \varepsilon(x)|\phi'(x)|^2 dx, \quad \phi \in W^{1,2}_0, 
\end{displaymath}
what implies
\begin{displaymath}
\|\phi\|^2_ {W^{1,2}_0}\le \|2/\varepsilon\|_{L^\infty} <{\mathcal P}\phi,\phi>
\end{displaymath}
for $\phi \in W^{1,2}_0$. Hence, by the Lax-Milgram lemma, the inverse
operator ${\mathcal P}^{-1}: W^{-1,2}\longmapsto W^{1,2}_0$ exists and
its norm does not exceed $\|2/\varepsilon\|_{L^\infty}$.
\begin{assu}\label{assu-dot}
{\rm
The density of ionized dopants $D$ is a 'real distribution' from
$W^{-1,2}$, what means that it takes real values if applied to real elements from
$W^{1,2}_0$.
}
\end{assu}
\begin{defn}
Let Assumptions \ref{assu-diel} and \ref{assu-dot} be
satisfied. Further, suppose $u \in L^1 \hookrightarrow W^{-1,2}$. We say 
that $\varphi \in W^{1,2}$ is a solution of the Poisson equation (\ref{e-poi}) obeying
the inhomogeneous Dirichlet boundary conditions (\ref{e-bv}) if
$\phi := \varphi - \widetilde{\varphi} \in W^{1,2}_0$ satisfies
\begin{displaymath}
{\mathcal P}\phi = D - qu.
\end{displaymath}
We set 
\begin{equation}\label{4.6}
\varphi(u) := \widetilde{\varphi} + {\mathcal P}^{-1}(D - qu).
\end{equation}
\end{defn}

Next we are going to introduce the Schr\"odinger operator.
\begin{assu}\label{assu-effekt}
{\rm
The effective mass $m$ is a real, non-negative function obeying $m \in L^\infty$
and $\frac{1}{m} \in L^\infty$.
}
\end{assu}
\begin{defn}\label{d-hv}
{\rm
Let Assumption \ref{assu-effekt} be satisfied.  
If $V \in L^1$ is real valued, then
the Schr\"odinger operator $ H_V : W^{1,2}_0 \longmapsto W^{-1,2}$
corresponding to the potential $V$ is defined by
\begin{displaymath}
<  H_V v,w> 
    = \frac{\hbar^2}{2} \int_\A^\B
    \frac{1}{m(x)}\, v'(x)\, \overline{w}'(x) dx 
    + \int_\A^\B V(x)v(x)\overline{w}(x) dx , 
  \end{displaymath}
$v,w \in W^{1,2}_0$. The definition is justified, because  $W^{1,2}_0$ continuously
embeds into $\lun $. Thus, the second term on the right hand side 
is always finite and defines a continuous sesquilinear form
on $W^{1,2}_0$. The operator with zero potential will be denoted by $H_0$ in the sequel.
The restriction of the operators just introduced to other range
spaces, in particular, to  $\lzw$, we also denote by $H_V$.  
}
\end{defn}
For any real valued $V \in \lei $, the restriction of $H_V$
to the range space $\lzw$ is selfadjoint and has a complete orthonormal 
system of eigenfunctions. All eigenvalues are then real and simple. 
\begin{defn} \label{d-vert}
{\rm
We say that a continuous function $f:\R \mapsto [0,\infty[$ is from the
class $\mathcal D$, if one of the following conditions is satisfied:
\begin{enumerate}

\item[(i)]
there is a $t \in ]-\infty,\infty[$ such that $f$ is strictly monotonously decreasing on
the interval $]-\infty,t[$ and zero on $[t,\infty[$,

\item[(ii)]
the function $f$ obeys $f(s) > 0$ for any $s \in \R$ and
\end{enumerate}
\begin{displaymath}
\sup_{s \in [1,\infty[}f(s)s^2 <\infty.
\end{displaymath}
is valid.
}
\end{defn}
\begin{assu}\label{assu-fermi}
{\rm
In the sequel we assume that all occurring distribution 
functions are from the space $\mathcal D$.
}
\end{assu}
\begin{rem} \label{r-temp0}
{\rm
Let us explicitely note that in contrast to other papers besides continuity no
further regularity assumptions are 
imposed on the distribution functions. Only this allows to include the zero
temperature case, see Chapters~6.
}
\end{rem}
\begin{lem} \label{l-feermi}
Let Assumptions \ref{assu-effekt} and \ref{assu-fermi} be satisfied.
If $\{\lambda_l(V)\}_l$ are the eigenvalues of the Schr\"odinger operator
$H_V$ with real potential $V \in L^1$, 
then for every $N \in [1,\infty[$ there is exactly one real number $E \in \R$ which satisfies
\begin{displaymath}
{2}\sum_l f(\lambda_l(V)-E)=N,
\end{displaymath}
This real number is called the chemical potential and is denoted by $\mu_f(V)$.
\end{lem}
\begin{proof}
For every $f \in \mathcal{D}$ and every $E \in \R$ the expression 
$\sum_l f(\lambda_l(V)-E)$ is finite, 
$\lim\limits_{E \mapsto -\infty}\sum_l f(\lambda_l(V)-E)=0$ and 
$\lim\limits_{E \mapsto \infty}\sum_l f_l(\lambda(V)-E)=\infty. $ 
Moreover, it is not hard to see 
that on the set $\{E \in \R: \sum_l f(\lambda_l(V)-E) \neq 0 \}$ 
the function $E \rightarrow 
\sum_l f(\lambda_l(V)-E)$
is strictly increasing.
\end{proof}
\begin{rem}\label{r-fixx}
{\rm 
In the following we assume that the number of particles in the
Kohn-Sham system  is always fixed by $N \in [1,\infty[$ without indicating
this explicitly.
}
\end{rem}
\begin{defn} \label{d-particleop}
{\rm
Let Assumptions \ref{assu-effekt} and \ref{assu-fermi} be
satisfied. Further, 
let $\{\psi_l(V)\}_l$ and $\{\lambda_l(V)\}_l$ be the eigenfunctions
and eigenvalues of the Schr\"o\-dinger operator
$H_V$ with real potential $V \in L^1$ and let 
$\mu_f(V)$ be the chemical potential. 
Then
\begin{displaymath}
    \mathcal{N}_f(V)(\fx ) 
    := {2}\sum_{l}
    f(\lambda_{l}(V)-\mu_f(V))\big| \psi_l(V)(\fx) \big|^2, \quad \fx\in [\A,\B]
\end{displaymath}
defines an operator $\mathcal{N}_f: L^1 \longmapsto L^1$ which is called the particle density operator.
}
\end{defn}
\begin{rem}
{\rm
$\mathcal{N}_f(V)$ obviously satisfies
\begin{displaymath}
\int_0^1 \mathcal{N}_f(V)dx
= {2}\sum_{l=1}^\infty f(\lambda_l(V)-\mu_f(V))=N.
\end{displaymath}
%
%
%
%
}
\end{rem}
\begin{assu}\label{assu-xc}
{\rm
a) The potential $\Delta E$ is a real-valued $L^1$ function.\\ 
b) The exchange--correlation term in its dependence on the 
  particle densities, i.e.\ the mapping $u \longmapsto V_{xc}(u)$, is a continuous 
and bounded mapping from $L^1$ into $L^1 $. This assumption covers
the Hartree-Fock-Slater and Thomas-Fermi exchange--correlation terms.
}
\end{assu}
\begin{defn}\label{IV.16}
{\rm
Let Assumptions \ref{assu-diel}, \ref{assu-dot},
\ref{assu-effekt}, \ref{assu-fermi} and \ref{assu-xc} be satisfied.
The pair $\{\varphi,u\} \in W^{1,2} \times L^1$ is called a
solution of the Kohn-Sham system with distribution function $f$, if $\varphi$ solves
 the Poisson equation (\ref{e-poi}) with inhomogeneous Dirichlet boundary conditions
(\ref{e-bv}) and the particle density $u$ is given by 
\begin{displaymath}
u := \mathcal{N}_f(\Delta E + V_{xc}(u) - q\varphi).
\end{displaymath}
}
\end{defn}

\section{Existence of solutions}\label{existence-chap}

In this section we are going to show that the Kohn-Sham system always admits  a solution. 
As in \cite{KR1,KR2} Schauder's fixed point theorem is used, what requires several 
estimates (e.g. eigenvalues of the Schr\"odinger operator, $W^{1,2}$-bounds of the 
particle density operator).
To assure its applicability, we first establish some prerequisites.
\begin{defn}
{\rm
For $m\in L^\infty$ we set
$\underline{m} :=\max\left(1,\frac{2\|m\|_{L^\infty}}{\hbar^2}\right)$.
}
\end{defn}
\begin{rem}\label{m-remark}
{\rm
Recognizing that $\underline{m}$ has been defined such that $1/\underline{m}$ is a
monotonicity constant of the operator $H_1=H_0+1:W_0^{1,2}\mapsto W^{-1,2}$,
the Lax-Milgram lemma shows, that the norm of the inverse operator is not larger
than $\underline{m}$,
\begin{displaymath}
\|\psi\|^2_{W^{1,2}_0}\leq\underline{m}<(H_0+1)\psi,\psi
>,\quad\left\|(H_0+1)^{-1}\right\|_{\mathcal{B}(W^{-1,2},W_0^{1,2})}
\le \underline{m}.
\end{displaymath}
}
\end{rem}
By some calculations one finds
\begin{displaymath}
\|\psi\|_{L^\infty}\leq\sqrt{2}\, 
\|\psi\|_{W^{1,2}_0}^{\frac{1}{2}}\|\psi\|_{L^2}^{\frac{1}{2}}\,,
\end{displaymath}
which proves the continuous embedding $W^{1,2}_0\hookrightarrow L^{\infty}$.

The following proposition allows to compare the eigenvalues of the Schr\"o\-dinger operators
$H_V$ and $H_0$ and, additionally, provides a comparison between the operators 
$(H_0+1)^{-1/2}$ and $(H_V-\rho)^{-1/2}$. Both we will need later on as technical 
instruments.
\begin{prop}\label{H0-HV-prop}
Let Assumption \ref{assu-effekt} be satisfied and let $V \in
L^1$ be real-valued. Then the following holds:
\begin{enumerate}
\item[{\rm(i)}] The eigenvalues $\lambda_l(V)$ of the 
operator $H_V$ can be estimated as follows:
\begin{equation}\label{EW absch.}
\frac{1}{2}(\varsigma_l+1)+\rho_V\leq \lambda_l(V)\leq\frac{3}{2}(\varsigma_l+1)-\rho_V-2,
\quad l=1,2,\ldots
\end{equation}
where the $\varsigma_l$ are the eigenvalues of the operator $H_0$ and
$\rho_V$ is given by
\begin{displaymath}
\rho_V=-2\|V\|_{L^1}^{2}\underline{m}-1\,.
\end{displaymath}

\item[\rm{(ii)}]
For $\rho\leq\rho_V$ the spectrum of $(H_V-\rho)^{-1}$ is contained in
$[0,2]$ and 
\begin{equation}\label{HV-H0-ineq}
\left\|(H_V-\rho)^{-\frac{1}{2}}(H_0+1)^{\frac{1}{2}}\right\|=
\left\|(H_0+1)^{\frac{1}{2}}(H_V-\rho)^{-\frac{1}{2}}\right\|\leq\sqrt{2}\,.
\end{equation}
\end{enumerate}
\end{prop}
A proof of this is to be found in \cite{KR1}, see Prop.~3.3.
\begin{rem} \label{r-unifbound}
{\rm
The form which defines $H_0$ may be estimated as follows:
\begin{eqnarray*}
\lefteqn{
\frac{\hbar^2}{2} 
\essinf\limits_{x\in(0,1)}\left\{\frac{1}{m}\right\}\int_\A^\B \, |v'(x)|^2dx \le}\\
& &  
\frac{\hbar^2}{2} \int_\A^\B m(x)^{-1}\, |v'(x)|^2 dx \le 
\frac{\hbar^2}{2} 
\esssup\limits_{x\in(0,1)}\left\{\frac{1}{m}\right\}\int_\A^\B\,|v'(x)|^2 dx.
\end{eqnarray*}
Thus, the eigenvalues of $H_0$ compare by the minimax principle from below and above
in an obvious manner with the eigenvalues of 
$-\frac {d^2}{dx^2}$  -- combined with a {homogeneous} Dirichlet condition.\\
The reader should notice that \proref{H0-HV-prop} gives uniform bounds with respect to
$L^1-$bounded sets of potentials. 
}
\end{rem}
From  \proref{H0-HV-prop} we can deduce the following
\begin{lem}\label{lem-resolvent-continuity}
Let Assumption \ref{assu-effekt} be satisfied and let 
$\mathcal{M}\subset L^1$ be a bounded set of real-valued potentials. If 
\begin{displaymath}
\rho < \rho_\mathcal{M} :=\inf_{V\in\mathcal{M}}\rho_V=-2\underline{m}
\sup_{V\in\mathcal{M}}\|V\|^2_{L^1}-1, 
\end{displaymath}
then the mapping 
\begin{displaymath}
L^1 \supset \mathcal{M} \ni V\mapsto (H_V-\rho)^{-1} \in \mathcal{B}(L^2)
\end{displaymath}
is Lipschitz continuous in $V$ with a Lipschitz constant depending on
$\mathcal{M}$.
\end{lem}
\begin{proof}
If $\rho \le \rho_\mathcal{M}$, then $\rho$ belongs to the resolvent
set of $H_V$ for any $V \in \mathcal{M}$ by Proposition \ref{H0-HV-prop}. 
Moreover, one has
\begin{equation} \label{e-lowbound}
H_V - \rho =  H_V - \rho_V + \rho_V - \rho \ge
\frac{1}{2}I
\end{equation}
 since  $H_V - \rho_V \ge \frac{1}{2}I$ for $V \in
\mathcal{M}$ by Proposition \ref{H0-HV-prop}.
Applying the resolvent equation
\begin{displaymath}
(H_V - \rho)-(H_U - \rho)=(H_V - \rho)(U-V)(H_U - \rho),
\end{displaymath}
one obtains
\begin{eqnarray*}
\lefteqn{
\|(H_V - \rho)^{-1}-(H_U - \rho)^{-1}\|=}\\
& &
\|(H_V - \rho)^{-1}(U-V)(H_U - \rho)^{-1}\| .
\end{eqnarray*}
The latter term may be estimated as follows:
\begin{eqnarray*}
\lefteqn{
\|(H_V - \rho)^{-1}(U-V)(H_U - \rho)^{-1}\| \le
\|U-V\|_{\mathcal{B}(L^\infty;L^1) }\,\times
}\\
& &
\|(H_V - \rho)^{-1/2}\|\|(H_V - \rho)^{-1/2}(H_0+1)^{1/2}\| \times\\
& &
\|(H_0+1)^{-1/2}\|_{\mathcal{B}(L^1;L^2)}
\|(H_0+1)^{-1/2}\|_{\mathcal{B}(L^2;L^\infty)}\times\\
& &
\|(H_0+1)^{1/2}(H_U - \rho)^{-1/2}\| \|(H_U - \rho)^{-1/2}\|.
\end{eqnarray*}
The factors $\|(H_V - \rho)^{-1/2}\|$ and $\|(H_U - \rho)^{-1/2}\|$ are uniformly bounded
 in $U,V$ due to (\ref{e-lowbound}). $\|(H_V - \rho)^{-1/2}(H_0+1)^{1/2}\|$,
$\|(H_0+1)^{1/2}(H_U - \rho)^{-1/2}\|$ are uniformly bounded by (\ref{HV-H0-ineq}). 
Furthermore, $\|(H_0+1)^{-1/2}\|_{\mathcal{B}(L^2;L^\infty)}$ is finite
by the embedding $W_0^{1,2} \hookrightarrow L^\infty$. This implies 
$\|(H_0+1)^{-1/2}\|_{\mathcal{B}(L^1;L^2)} <\infty$  by duality. Finally,
$\|U-V\|_{\mathcal{B}(L^\infty;L^1)}$ is identical with $\|U-V\|_{L^1}$.
\end{proof}
\begin{cor} \label{c-eigenva}
Let Assumption \ref{assu-effekt} be satisfied. 
If $\{V_n\}_n$ converges to $V$ in $L^1$, then the operator
sequence $\{H+V_n\}_{n\in\N}$ converges in the norm resolvent sense to $H+V$.
In particular, the eigenvalues and the eigenprojections of 
$H+V_n$ converge to the corresponding eigenvalues and eigenprojections
of $H+V$.
\end{cor}
The proof follows from the preceding corollary and a well known perturbation theorem, see
\cite[Ch.~IV.3.4]{Ka1}.
\begin{lem} \label{l-fermi}
Let Assumptions \ref{assu-effekt} and \ref{assu-fermi} be satisfied. 
\begin{enumerate}

\item[{\rm (i)}]
For any bounded set of real potentials in $L^1$ the set of
chemical potentials is also bounded. 
Additionally, this bound can be taken even uniform with respect to any subset $\mathcal{D}_1
\subset \mathcal{D}$ of distribution functions $f$ obeying in addition
\begin{equation} \label{e-uniffermi}
\sup_{f \in \mathcal{D}_1} \sup_{t \in [1,\infty[}f(t)t < \infty
\quad \text{ and} \quad
\inf_{f\in \mathcal{D}_1} f(a) >0 \quad \text{for one} \quad a \in \R.
\end{equation}

\item[{\rm (ii)}]
Let $\mathcal{C} := \{f_j \}^\infty_{j=1}$ be a sequence of functions
from $\mathcal{D}_1$ which converges 
uniformly on  bounded intervals to a function $f \in \mathcal{D}$ as $j \to \infty$.  
If $V_j \mapsto V$ in $L^1$, then  
$\lim\limits_{j\to\infty}\mu_{f_j}(V_j) = \mu_f(V)$.
\end{enumerate}
\end{lem}
\begin{proof}
(i) By the monotonicity of $f$, 
one has $\sum_l f(\lambda _l(V) - E) \ge k \inf_{f\in \mathcal{C}} f(a)$
if $k $ items $\lambda _l - E$ are situated below $a$. This means
$\sum_l f(\lambda _l(V) - E) \mapsto \infty$ for $E \mapsto \infty$ uniform within
$\mathcal{C}$. Thus, the chemical potential $\mu_f(V)$ have to be bounded from
above by the monotonicity of $f$ and (\ref{EW absch.}), uniformly within the class $\mathcal{C}$. 
On the other hand, if $E \le \inf_l \inf \spec(H_{V_l})-1$, then one
can estimate
\begin{displaymath}
 \sum_l f(\lambda _l(V) - E) \le  \sup_{f \in \mathcal{C}} \sup_{t \in [1,\infty[}f(t)t
\sum_l (\lambda _l(V) - E)^{-1},
\end{displaymath}
what --again by (\ref{EW absch.}) -- tends to zero for $E \mapsto -\infty$ uniformly for
 $f \in \mathcal{C}$ and uniform over sets of potentials which are
 bounded in $L^1$.
 
(ii) First, the chemical potentials are uniformly bounded, due to (i). Thus, the eigenvalues
of the operators $H_{V_j-\mu_{f_j}(V_j)}$ admit uniform bounds as in \proref{H0-HV-prop}.
Assume that the assertion was not true; then for a subsequence
$\{V_k\}_k$ one has
\begin{equation} \label{e-wid1}
|\mu_{f_k}(V_k)-\mathcal E_f(V)|> \epsilon>0.
\end{equation}
Because the chemical potentials form a bounded set, we may again pass to a subsequence $\{V_n\}_n$,
and suppose
\begin{equation} \label{e-wid2}
\lim_{n \rightarrow \infty} \mu_{f_n}(V_n)= E \neq \mu_f(V).
\end{equation}
We will now show that this leads to 
\begin{equation} \label{e-wid12}
\lim_{ n \mapsto \infty} \sum_l f_n(\lambda_{l}(V_n)-\mu_{f_n}(V_n)) = 
\sum_l f(\lambda_{l}(V)-E).
\end{equation}
First, it follows from (\ref{e-uniffermi}) and \proref{H0-HV-prop} that for any $\delta >0$ 
there is a number $l_0$ such that 
\begin{equation} \label{e-wid3}
 \sum_{l>l_0} f_n(\lambda_{l}(V_n)-\mu_{f_n}(V_n))< \delta \quad \text {and} \quad
\sum_{l>l_0} f(\lambda_{l}(V)-E) < \delta
\end{equation}
uniformly for all $n$. 
The remaining eigenvalues $\lambda_l(V_n)$ for $l \le l_0$ and all $n$ lie in a bounded
interval and, additionally, due to \corref{c-eigenva}, 
one has $\lim\limits_{n \mapsto \infty} \lambda_l(V_n)
=\lambda_l(V)$ for every $l$. Thus, according to (\ref{e-wid2}), (\ref{e-wid3}) and the
uniform convergence of $\{f_j\}_j$ to $f$ on bounded intervals as $j
\to \infty$, the term
\begin{displaymath}
\bigl | \sum_l f_n(\lambda_{l}(V_n)-\mu_{f_n}(V_n)) - 
\sum_l f(\lambda_{l}(V)-E)\bigr |
\end{displaymath}
becomes smaller than $3 \delta$ for sufficiently large $n$ and arbitrarily chosen $\delta >0$.
 This implies (\ref{e-wid12}),
but (\ref{e-wid12}) must be false: the terms on the left hand side all equal $N$, what cannot
be true for the right hand side due to (\ref{e-wid2}) and
\lemref{l-feermi}. Hence, (\ref{e-wid1})
is wrong, what proves (ii).
\end{proof}
\begin{rem} \label{r-fermicont}
{\rm 
The lemma shows in particular that the chemical potentials continuously
depend on the potential.
}
\end{rem}
\begin{thm}\label{t-W12-absch}
Let Assumptions \ref{assu-effekt} and \ref{assu-fermi} be
satisfied and let $\mathcal{M}$ be a bounded set of real potentials in $L^1$.
Then the following holds:
\begin{enumerate}

\item[{\rm (i)}]
The image $\mathcal{N}_f(\mathcal{M})$, $f \in \mathcal{D}$, is a bounded set in $W_{0}^{1,2}$.
The bound may be taken uniformly with respect to any 
set $\mathcal{D}_2 \subset \mathcal{D}$ of distribution functions $f$
which satisfy the additional conditions
\begin{equation} \label{e-uniffermi00}
\sup_{f \in \mathcal{D}_2} \sup_{t \in [0,\infty[}f(t)t(t+1) < \infty
\quad \text{ and} \quad
\inf_{f\in \mathcal{D}_2} f(a) >0 \quad \text{for one} \quad a \in \R.
\end{equation}

\item[{\rm (ii)}]
The particle density operator $\mathcal{N}_f:L^1 \mapsto
L^1$, $f \in \mathcal{D}$, is continuous.
\end{enumerate}
\end{thm}
\begin{proof}
(i) For $V\in\mathcal{M}$ we get
\begin{eqnarray}\label{W12-absch-topol-alg}
\lefteqn{
\|\mathcal{N}_f(V)\|_{W^{1,2}_{0}}=
\|{2}\sum_{l}f(\lambda_l(V)-\mu_f(V))|\psi_l|^2\|_{W_{0}^{1,2}}\leq}\\
& & 
{2}\sum_{l}f(\lambda_l(V)-\mu_f(V))\||\psi_l|^2\|_{W_{0}^{1,2}} \leq
4\sum_{l}f(\lambda_l(V)-\mu_f(V)) \|\psi_l\|_{W^{1,2}_0}^2
\nonumber
\end{eqnarray}
where in the last step we used
$\||\psi|^2\|_{W^{1,2}_0}\leq 2 \|\psi\|^2_{W^{1,2}_0}$ for all $\psi\in W^{1,2}_0$.
We estimate the terms in (\ref{W12-absch-topol-alg}):
\begin{eqnarray*}
\lefteqn{
 \|\psi_l\|^2_{W_{0}^{1,2}} \le 
\underline{m} \|(H_0+1)^{\frac{1}{2}}\psi_l\|^2_{L^2} \le }\\
& & \underline{m} \|(H_0+1)^{\frac{1}{2}}
(H_V-\rho)^{-\frac{1}{2}}\|^2\|(H_V-\rho)^{\frac{1}{2}}\psi_l\|^2_{L^2}
\leq 2\underline{m}  (\lambda_l(V)-\rho),
\end{eqnarray*}
(see (\ref{HV-H0-ineq})) where $\rho < \rho_\mathcal{M}$ where
$\rho_\mathcal{M}$ given by Lemma \ref{lem-resolvent-continuity} is a uniform lower bound for the
spectra of the operators $H_V$ with $V \in \mathcal{M}$. 
Thus, (\ref{W12-absch-topol-alg}) is not larger than 
{\small
\begin{eqnarray*}
\lefteqn{\|\mathcal{N}_f(V)\|_{W^{1,2}_{0}} \le
8 \underline{m} \sum_{l}f(\lambda_l(V)-\mu_f(V))(\lambda_l(V)-\rho)
 \le}\\
& &
 8 \underline{m}
 \sum_{l}f(\lambda_l(V)-\mu_f(V))\left[(\mu_f(V)-\rho) +
 (\lambda_l(V)-\mu_f(V))\right] \le\\
& &
8 \underline{m} N |\mu_f(V)-\rho| +
8 \underline{m}\sum_{l: \,\lambda_l(V) -  \mu_f(V) \ge 0}f(\lambda_l(V)-\mu_f(V))
(\lambda_l(V)-\mu_f(V)). 
\end{eqnarray*}
}
The last sum may be estimated by 
\begin{displaymath}
\sup_{t \in [0,\infty[} f(t)t(t+1)
\sum_{l: \lambda_l(V) -  \mu_f(V) \ge 0}(\lambda_l(V)-\mu_f(V)+1)^{-1}. 
\end{displaymath}
Obviously, the condition (\ref{e-uniffermi00}) is stronger than (\ref{e-uniffermi});
thus, the set of chemical potentials  is uniformly bounded for $f \in \mathcal{C}$ and 
$V \in \mathcal{M}$. This, together with the eigenvalue estimates (\ref{EW absch.})
and \remref{r-unifbound} proves (i).

(ii) According to \corref{c-eigenva} and \lemref{l-fermi} the eigenvalues and the corresponding
chemical potentials  depend continously on $V \ni L^1$.
 Furthermore, if $\{V_j\}_j$ converges in $L^1 $ to $V$, then -- due to \corref{c-eigenva} --
the (one-dimensional) eigenprojections $P_{j,k}=\langle \cdot,f_{j,k}\rangle f_{j,k}$ 
converge to the corresponding eigenprojections $P_k=\langle \cdot,f_k\rangle f_k$.
Applying these eigenprojections to the vector $f_k$, one easily
obtains 
\begin{displaymath}
|\langle f_k,f_{j,k}\rangle| \mapsto 1 \quad \text{and} \quad \langle f_k,f_{j,k}\rangle f_{j,k}
\mapsto f_k \quad
\text{in}\; L^2  \quad \text{ for} \quad  j \mapsto \infty.
\end{displaymath}
From this it is not hard to see that $|f_{j,k}|^2 \mapsto |f_k|^2$ in $L^1$.
Observing that for sufficiently large $M$
\[
\|\sum_{l \ge M}f(\lambda_{l}(V)-\mu_f(V))\,|\psi_l(V)|^2\|_{L^1} \le
\sum_{l \ge M}f(\lambda_{l}(V)-\mu_f(V))
\]
can be made arbitrarily small uniformly over an $L^1$-bounded set of potentials $V$, one 
obtains (ii).
\end{proof}
\begin{cor} \label{c-improvecont}
Let Assumptions \ref{assu-effekt} and \ref{assu-fermi} be satisfied.
If $X$ is a Banach space continuously injecting into $L^1$ such that $W^{1,2}_0$
compactly embeds into $X$, 
then the particle density operator $\mathcal{N}_f:L^1 \mapsto X$ is well defined and
continuous.
\end{cor}
\begin{proof}
That it is well defined follows from the embedding $W^{1,2}_0 \hookrightarrow X$.
Assume that the continuity does not hold; then there must be a sequence 
$\{V_n\}_n$ converging in $L^1$ to $V$ such that
\begin{equation} \label{e-contrad}
\|\mathcal{N}_f(V_n)-\mathcal{N}_f(V)\|_X> \epsilon >0.
\end{equation}
The statement (i) of the foregoing lemma tells us that
$\{\mathcal{N}_f(V_n)\}_n$  is bounded in $W^{1,2}$. 
Thus, by the compactness of the embedding $W^{1,2} \hookrightarrow X$
there must be a subsequence $\{V_k\}_k$ such that $\mathcal{N}_f(V_k)$ converges in $X$ 
to an element $\mathcal W$. But $\mathcal{N}_f(V_k)$ converges to $\mathcal{N}_f(V)$ in $L^1$.
Thus $\mathcal{W}$ must equal $\mathcal{N}_f(V)$ by the continuous
injection $X \hookrightarrow L^1$, what contradicts (\ref{e-contrad}).
\end{proof}

To prove existence of solutions of the Kohn-Sham system we will
 introduce an appropriate subset of $L^1$ together with a suitable mapping $\Phi$ from this set
 into itself. $\Phi$ will be constructed such that the solutions to the Kohn-Sham system
 coincide with the fixed points of this mapping.
\begin{defn}\label{def-FP-mapping}
{\rm
We set 
\begin{displaymath}
L^1_N :=\left\lbrace {u}\in  L^1: u\geq0,\quad \int^1_0 u \;dx=N \right\rbrace.
\end{displaymath}
Let 
Assumptions \ref{assu-diel}, \ref{assu-dot}, \ref{assu-effekt}, \ref{assu-fermi} and 
\ref{assu-xc} be satisfied. Then we define the mapping $\Phi_f:L^1_N\rightarrow L^1_N$ by
\begin{equation}\label{5.15}
\Phi_f({u}) :=\mathcal{N}_f(\Delta E+V_{xc}({u}) - q\varphi({u}))
\end{equation}
where $\varphi({u})$ denotes the ${u}$-dependent solution to Poisson's equation
given by (\ref{4.6}).
}
\end{defn}

The task now is to verify that $\Phi_f$ has a fixed point. We will use Schauder's theorem to 
achieve this.
\begin{thm}\label{V.12}
Let 
Assumptions \ref{assu-diel}, \ref{assu-dot}, \ref{assu-effekt}, \ref{assu-fermi} and 
\ref{assu-xc} be satisfied. 
Then
\begin{enumerate}

\item[$\rm (i)$]
the pair $\{\varphi(u),u\}$ is a solution of 
the Kohn-Sham system if and only if $u \in L^1_N$  is a fixed
point of $\Phi_f$, where $\varphi(u)$ is given by (\ref{4.6})\\
 and 

\item[$\rm (ii)$]
the mapping $\Phi_f$ has a fixed point.

\end{enumerate}
\end{thm}
\begin{proof}
(i) The first part of theorem follows immediately 
from the definition of the mappings $\mathcal{N}_f$ and $\Phi_f$.

(ii) To prove the second part we note that $L^1_N$ is a closed,
bounded and convex set, which is mapped by $\Phi_f$ into itself
 (Def. \ref{def-FP-mapping}).

\textit{Continuity:} Since $L^1\hookrightarrow W^{-1,2}$, the solution $\varphi$ to
 Poisson's equation depends continuously (in $W^{1,2}_0$) on $u$. Hence, due to 
\assuref{assu-xc}, the mapping $L^1 \ni  
u\mapsto \Delta E+V_{xc}(u)-q\varphi(u)
\in L^1$ is continuous. \thmref{t-W12-absch} then implies the continuity of $\Phi_f$.

\textit{Compactness:} According to Theorem \ref{t-W12-absch}, the image of a $L^1$-bounded set
of potentials is bounded in the space $W^{1,2}$. The compactness of the embedding 
$W^{1,2}\hookrightarrow L^1$ yields the compactness of $\Phi_f$.

Thus, due to Schauder's fixed point theorem, $\Phi_f$ has a fixed
point in $L^1_N$.
\end{proof}

\section{Particle density operator and monotonicity}

In this section we want to show some additional properties of the particle density operator.
For this it is necessary to restrict some considerations to the real parts of spaces
which were up to now considered as complex ones. We indicate this by an additional
subscript $\R$, e.g. $L^\infty_\R$.
The upcoming results are based on the following theorem.

\begin{thm}\label{allg-trace-satz}
Let $H$ be a self-adjoint operator in the separable Hilbert space $\mathcal{H}$ with
 compact resolvent and let $U$ and $V$ be bounded, self-adjoint operators on 
$\mathcal{H}$. If $f:\mathbb{R}\longrightarrow\mathbb{R}$ is a Borel 
measurable function such that $f(H+U)$
and $f(H+V)$ are trace class operators, then the formula
\begin{eqnarray}\label{allg-trace-gl}
\lefteqn{
\tr([f(H+U)-f(H+V)](U-V))=}\\
& &
\sum_{k,l=1}^{\infty}(f(\lambda_k)-f(\mu_l))(\lambda_k-\mu_l)
|\langle \psi_k,\xi_l\rangle|^2
\nonumber
\end{eqnarray}
is valid. Here $\{\lambda_k\}_k$ $(\{\mu_l\}_l)$ is the sequence of eigenvalues of 
$H+U$ $(H+V)$ and $\{\psi_k\}_k$ $(\{\xi_l\}_l)$ is an -- orthornormalized -- sequence of
corresponding eigenvectors. 
\end{thm}
The proof is given in \cite{KNR1}.
\begin{cor} \label{c-monoton}
If Assumptions \ref{assu-effekt} and \ref{assu-fermi} are satisfied,
then the mappings $L^1_\R  \ni V \rightarrow \mathcal{N}_f(V) \in L^\infty_\R$ and
$W^{1,2}_{0,\R} \ni V \rightarrow \mathcal{N}_f(V) \in  W^{-1,2}_{\R}$
are monotone.
\end{cor}
\begin{proof}
We specify the Hilbert space $\mathcal{H}$ to $L^2$ and first consider potentials 
$U,V \in L^\infty_\R$, which are identified with the induced
multiplication operators on $L^2$. Replacing $U,V$ in the
preceding theorem by $U-\mu_f(U), V-\mu_f(V)$ and observing
\begin{eqnarray*}
& &\tr\bigl ([f(H+U-\mu_f(U))-f(H+V-\mu_f(V))]
(\mu_f(U)-\mu_f(V))\bigr ) =\\
& &\tr\bigl ([f(H+U-\mu_f(U))-f(H+V-\mu_f(V))]\bigr)
(\mu_f(U)-\mu_f(V))=0
\end{eqnarray*}
one obtains 
\begin{eqnarray*}
\lefteqn{
\int_0^1 \bigl(\mathcal{N}_f(U)-\mathcal{N}_f(V)\bigr)(U-V) \,dx =}\\
& & \hspace{-5mm}
{2}\tr\bigl ([f(H+U-\mu_f(U))-f(H+V-\mu_f(V))](U-V)\bigr) =
\nonumber \\
& & \hspace{-5mm}
{2}\tr\bigl ([f(H+U-\mu_f(U))-f(H+V-\mu_f(V))]
(U-\mu_f(U)-V+\mu_f(V))\bigr).
\nonumber
\end{eqnarray*}
In view of (\ref{allg-trace-gl}) the right hand side is negative because the distribution
function $f$ is monotonously decreasing.
 
Let now $U,V \in L^1_\R$ be arbitrary. If $\{U_n\}_n, \{V_n\}_n$ are sequences from
 $L^\infty_\R$ which converge to $U,V$ in $L^1$, respectively,
 then, on one hand, $\mathcal{N}_f(U_n) \mapsto 
\mathcal{N}_f(U)$ and $\mathcal{N}_f(V_n) \mapsto \mathcal{N}_f(V)$ in $L^\infty_\R$, due to 
\corref{c-improvecont}. Because one already knows that 
$\int_0^1 \bigl(\mathcal{N}_f(U_n)-\mathcal{N}_f(V_n)\bigr)(U_n-V_n) \,dx \le 0$, one obtains
$\int_0^1 \bigl(\mathcal{N}_f(U)-\mathcal{N}_f(V)\bigr)(U-V) \,dx \le 0$. 

The $W^{1,2}_{0,\R} \leftrightarrow W^{-1,2}_{\R}$ duality extends the $L^\infty_\R
 \leftrightarrow L^1_\R $ duality; thus, the second assertion follows from the first. 
\end{proof}
\begin{thm} \label{c-unique}
Let 
Assumptions \ref{assu-diel}, \ref{assu-dot},
\ref{assu-effekt} \ref{assu-fermi} and \ref{assu-xc} a) be satisfied.
If the correlation term $V_{xc}$ is absent, then the Kohn-Sham system
has a unique solution $\{\varphi,u\}$. 
\end{thm}
\begin{proof}
It is not hard to see that in this case the system can be written as one equation
for the electrostatic potential in the real space $W^{-1,2}_{\R}$, namely 
\begin{equation} \label{e-keinVxc}
-\frac {d}{dx}\varepsilon \frac {d}{dx} \phi +q\mathcal{N}_f(\Delta E
- q\widetilde{\varphi} - q\phi)= D
\end{equation}
where $\varphi = \widetilde{\varphi} + \phi$. Since the operator $-\frac {d}{dx}\varepsilon
 \frac {d}{dx}:W^{1,2}_{0,\R} \mapsto W^{-1,2}_{\R}$ is strongly monotone, the operator  
$-\frac {d}{dx}\varepsilon \frac {d}{dx} +q \mathcal{N}_f(\Delta E-q\widetilde{\varphi} - q\;\cdot):W^{1,2}_{0,\R} \mapsto W^{-1,2}_{\R}$ is
also strongly monotone by the foregoing corollary. Additionally, this latter operator is
continuous, and hence (\ref{e-keinVxc}) has a unique solution by the theory of monotone 
operators, see \cite{zeid} Ch.26.2.
\end{proof}
\begin{cor} \label{c-apriori}
Let the assumptions of Theorem \ref{c-unique} be satisfied. 
If $\{\varphi,u\}$ is a solution of the Kohn-Sham system, then the electrostatic potential 
$\phi = \varphi -\widetilde{\varphi}$ satisfies the following a priori estimate:
\begin{displaymath}
\|\phi \|_{W^{1,2}_0}\le \frac {1}{M} \Bigl (\|D\|_{W^{-1,2}_0}+  
\gamma_{L^1;W^{-1,2}} N q\Bigr ),
\end{displaymath}
where M is the monotonicity constant for the operator 
$-\frac {d}{dx} \varepsilon \frac {d}{dx}: W^{1,2}_{0,\R} \mapsto W^{-1,2}_{\R}$ and 
$\gamma_{L^1;W^{-1,2}}$ is the norm of the embedding operator $L^1 \hookrightarrow W^{-1,2}$.
\end{cor}
\begin{proof}
Clearly, $\phi$ satisfies the equation
\begin{displaymath}
-\frac {d}{dx}\varepsilon \frac {d}{dx} \phi \;
+ q\mathcal{N}_f(\Delta E+V_{xc}(u) - q\widetilde{\varphi} -q\phi)= D,
\end{displaymath}
which can be regarded as an equation in $W^{-1,2}_{\R}$, due to \assuref{assu-dot} 
and the fact that the dielectric permittivity matrix $\varepsilon$ has real entries, cf.
\assuref{assu-diel}.
Considering $u$ as fixed, the left hand side is a strongly monotone, continuous operator
when acting on $\phi$. We denote it for brevity by $\mathcal{A}$. Its monotonicity constant
is at least $M$.
Using the strong monotonicity, we may estimate:
\begin{displaymath}
\|\mathcal{A}\phi-\mathcal{A} 0\|_{W^{-1,2}_{\R}} \|\phi \|_{W^{1,2}_{0,\R}} 
\ge \big |< \mathcal{A}\phi-\mathcal{A}0,\phi >\big |
\ge M \|\phi \|^2_{W^{1,2}_{0,\R}},
\end{displaymath}
what leads to
\begin{eqnarray*}
\lefteqn{
\|\phi \|_{W^{1,2}_{0,\R}}\le 
\frac {1}{M} \Bigl (\|D\|_{W^{-1,2}_\R}+
q \|\mathcal{N}_f(\Delta E +V_{xc}(u)-q\widetilde{\varphi})\|_{W^{-1,2}_{\R}}\Bigr ) \le}
\\
& &
\le \frac {1}{M} \Bigl (\|D\|_{W^{-1,2}_{\R}}+ q
\gamma_{L^1;W^{-1,2}_{\R}} \| \mathcal{N}_f(\Delta E +V_{xc}(u) -q 
\widetilde{\varphi})\|_{L^1_\R}  \Bigr ).
\end{eqnarray*}
Obviously, $\| \mathcal{N}_f(\Delta E +V_{xc}(u) -q \widetilde{\varphi})\|_{L^1_\R}$ equals $N$,
what gives the assertion.
\end{proof}
\begin{rem} \label{r-aprioindep}
{\rm
Note that this estimate depends in no way on the distribution function $f$
(within the class of monotonously decreasing functions, of course).
}
\end{rem}

\section{Convergence to zero temperature}

In the following we introduce the function $\theta$ by defining
\begin{displaymath}
\theta(x)=\left\{
\begin{array}{r@{\quad:\quad}l}
1 & x \le 1\\
x & x > 1. 
\end{array}\right.
\end{displaymath}
We start with the following technical lemma.
\begin{lem}\label{VII.1}
Let Assumptions \ref{assu-effekt} and \ref{assu-fermi} be
satisfied. Further, let $\mathcal{C}:=\{f_j\}^\infty_{j=1} \subset
\mathcal{D}_2$, cf. Theorem \ref{t-W12-absch},
and $f \in \mathcal{D}$ such that
\begin{equation} \label{e-convweight}
\lim_{j \mapsto \infty}\sup_{x \in [a,\infty[}|f_j(x)-f(x)|\theta(x)=0
\end{equation}
holds for every $a \in ]-\infty,-1[$. If $\{V_j\}^\infty_{j=1}$, $V_j
\in L^1$ converges to the real potential $V \in L^1$ in $L^1$, then 
$\lim\limits_{j\to\infty}\mathcal{N}_{f_j}(V_j) = \mathcal{N}_f(V)$ in $L^1$.
\end{lem}
\begin{proof}
One has the estimate
\begin{eqnarray*}
\lefteqn{
\|\mathcal{N}_{f_j}(V_j)- \mathcal{N}_f(V_{j})\|_{L^1} =}\\
& &
\sup_{\|W\|_{L^\infty}=1} \big |\int \bigl (\mathcal{N}_{f_j}
(V_j)-\mathcal{N}_f(V_j)\bigr )W \,dx \big |=
\nonumber\\
& &
\sup_{\|W\|_{L^\infty}=1}\big |{2}\tr\left(\left(f_j(H_{V_j} -\mu_{f_j}(V_j))-
f(H_{V_j} -\mu_f(V_j))\right)W\right)
\big | \le
\nonumber \\
& &
{2}|f_{j}(H_{V_j} -\mu_{f_j}(V_j))- f(H_{V_j} -\mu_f(V_j))\|_{\mathcal B_1}\le
\nonumber \\
& &
{2}|\bigl (f_j(H_{V_j} -\mu_{f_j}(V_j))-f(H_{V_j} -\mu_{f_j}(V_j))\bigr )
(H_{V_j} -\rho)\|_{\mathcal B} \|(H_{V_j}-\rho)^{-1}\|_{\mathcal B_1},
\nonumber
\end{eqnarray*}
where $\rho  < \rho_\mathcal{C}$, cf. Lemma \ref{lem-resolvent-continuity}.
This leads to the estimate 
\begin{eqnarray*}
\lefteqn{
\|\mathcal{N}_{f_j}(V_j)- \mathcal{N}_f(V_{j})\|_{L^1} \le}\\
& &
{2}|\bigl (f_j(H_{V_j} -\mu_{f_j}(V_j))-f(H_{V_j} -\mu_{f_j}(V_j))\bigr )
(H_{V_j} -\rho)\|_{\mathcal B} 
\|(H_{V_j}-\rho)^{-1/2}\|^2_{\mathcal B_2}, 
\nonumber 
\end{eqnarray*}
Using (\ref{HV-H0-ineq}), one estimates the second factor by
\begin{eqnarray*}
\lefteqn{
\|(H_{V_j}-\rho)^{-1/2}\|_{\mathcal{B}_2} \le }\\
& &
\|(H_{V_j}-\rho)^{-1/2}(H_0 - \rho)^{1/2}\|_{\mathcal{B}}\|(H_0 - \rho)^{-1/2}\|_{\mathcal{B}_2}
\le \sqrt{2}\|(H_0 - \rho)^{-1/2}\|_{\mathcal{B}_2}.
\end{eqnarray*}
To estimate the first one, we write
\begin{eqnarray}
\lefteqn{
\|\bigl (f_j(H_{V_j} -\mu_{f_j}(V_j))-f(H_{V_j} -\mu_f(V_j))\bigr )
(H_{V_j} -\rho)\|_{\mathcal B} \le }\nonumber\\
& &
\sup_{t \in [\inf_j \spec (H_{V_j}),\infty[}
|(f_j(t -\mathcal E_{f_j}(V_j))-f(t -\mathcal E_f(V_j))|(t -\rho) \le
\nonumber\\
& &
\sup_{t \in [\inf_j \spec (H_{V_j}),\infty[}
|(f_j(t -\mu_{f_j}(V_j))-f(t -\mu_{f_j}(V_j))|(t -\rho) +
\label{e-absch000}\\
& &
\sup_{t \in [\inf_j \spec (H_{V_j}),\infty[}
|(f(t -\mu_{f_j}(V_j))-f(t -\mu_f(V_j))|(t -\rho).
\label{e-absch0000}
\end{eqnarray} 
The term (\ref{e-absch000}) converges to zero due to
(\ref{e-convweight}) and $\mathcal{C} \subseteq \mathcal{D}_2$. 

We consider the term (\ref{e-absch0000}). \lemref{l-fermi} yields that $\{\mu_{f_j}(V_j)\}^\infty_{j=1}$ 
and $\{\mu_f(V_j)\}^\infty_{j=1}$ converge to $\mu_f(V)$. Thus, the restriction of 
$f(\cdot -\mu_{f_j}(V_j))-f(\cdot -\mu_f(V_j))$ to finite intervals uniformly converges to
zero  by the continuity of $f$. On the other hand, $f$ decays at $\infty$ as $\frac {1}{t^2}$, thus,
 for large arguments $t$ the absolute value of $f(t -\mu_{f_j}(V_j))-f(t -\mu_f(V_j))$ becomes
 arbitrarily small uniformly in $j$. This altogether shows that (\ref{e-absch0000}) goes to zero.
\end{proof}
\begin{cor}\label{VII.2}
Let 
Assumptions \ref{assu-diel}, \ref{assu-dot}, \ref{assu-effekt}, \ref{assu-fermi} and 
\ref{assu-xc} be satisfied. Further, let $\mathcal{C}:=\{f_j\}^\infty_{j=1} \subset
\mathcal{D}_2$, cf. Theorem \ref{t-W12-absch}, and $f \in \mathcal{D}$ such that condition 
(\ref{e-convweight}) holds for every $a \in ]-\infty,-1[$. If $\{u_j\}^\infty_{j=1}$, $u_j
\in L^1_N$ converges to the real function $u \in L^1_N$ in $L^1$, then 
$\lim\limits_{j\to\infty}\Phi_{f_j}(u_j) = \Phi_f(u)$ in $L^1$.
\end{cor}
\begin{proof}
By definition (\ref{5.15}) we have
\begin{displaymath}
\Phi_{f_j}(u_j) := \mathcal{N}_{f_j}(\Delta E +V_{xc}(u_j) \; - q\varphi(u_j)),
\quad u_j \in L^1_N, \quad j = 1,2,\ldots.\; ,
\end{displaymath}
where
\begin{displaymath}
\varphi(u_j) = \widetilde{\varphi} + \mathcal{P}^{-1}(D - qu_j).
\end{displaymath}
Since $\lim\limits_{j\to\infty}\varphi(u_j) = \varphi(u) = 
\widetilde{\varphi} + \mathcal{P}^{-1}(D - qu)$
in $W^{1,2}$ one gets that $\lim\limits_{j\to\infty}\varphi(u_j) =
\varphi(u)$ in $L^1$. Therefore $\lim\limits_{j\to\infty}V_j = V$ in $L^1$, where
\begin{displaymath}
V_j := \Delta E +V_{xc}(u_j) - q \varphi(u_j)
\quad \mbox{and} \quad 
V :=  \Delta E +V_{xc}(u) - q\varphi(u).
\end{displaymath}
Applying Lemma \ref{VII.1}, we complete the proof.
\end{proof}
\begin{thm} \label{t-conv}
Let 
Assumptions \ref{assu-diel}, \ref{assu-dot},
\ref{assu-effekt}, \ref{assu-fermi} and \ref{assu-xc} be
satisfied. Further, let $f \in \mathcal{D}$ and 
$\mathcal{C}:=\{f_j\}^\infty_{j=1} \subseteq \mathcal{D}_2$,
cf. Theorem \ref{t-W12-absch} such that
$\mathcal{C}:=\{f_j\}^\infty_{j=1}$ obeys (\ref{e-convweight})
for every $a \in ]-\infty,-1[$. 
If $\{\{\varphi_j,u_j\}\}^\infty_{j=1}$ are solutions of the Kohn-Sham system with respect
 to the distribution function $f_j$, then there is a subsequence 
$\left\{\{\varphi_k,u_k\}\right\}^\infty_{k=1}$ which converges in $L^\infty \times L^1$ to a
 solution $\{\varphi,u\}$ of the Kohn-Sham system with distribution function $f$.
\end{thm}
\begin{proof}
By Theorem \ref{V.12} $\{\varphi_j,u_j\} \in W^{1,2} \times L^1$ is a solution of the
Kohn-Sham system with respect to the distribution function $f_j$ if
and only if $u_j = \Phi_{f_j}(u_j)$, $j = 1,2,\ldots$ and the corresponding
potential is given by $\varphi_j = \varphi(u_j) = \widetilde{\varphi}
+ \mathcal{P}^{-1}(D -  qu_j)$. According to \thmref{t-W12-absch} 
there are subsequences  $\{u_k\}^\infty_{k=1}$ and
$\{\varphi_k\}^\infty_{k=1}$ such that the following properties are satisfied:
\begin{enumerate}

\item[$\bullet$] The sequence  $\{u_k\}^\infty_k$ is bounded in
$W^{1,2}$, obeys $u_k \in L^1_N$ and converges in $L^1$ to an element $u \in L^1_N$.

\item[$\bullet$] The sequence of potentials $\{\varphi_k\}^\infty_{k=1}$ converges 
in $L^1$ and, additionally, weakly in $W^{1,2}$ to an element $\varphi$.\\[-2mm]
 
\end{enumerate}
By Theorem \ref{V.12} the pair $\{\varphi_k,u_k\}$ is a solution of
the Kohn-Sham system 
with distribution function $f_k$ if and only if $u_k$ is a fixed point of the map $\Phi_k$, i.e.
\begin{displaymath}
u_k = \Phi_{f_k}(u_k), \qquad k = 1,2,\ldots\;.
\end{displaymath}
and the potential $\varphi_k$ is given by
\begin{displaymath}
\varphi_k = \widetilde{\varphi} + \mathcal{P}^{-1}(D - qu_k).
\end{displaymath}
By $\lim\limits_{k\to\infty}u_k = u$ in $L^1$ and Corollary \ref{VII.2}
we get $u = \Phi_f(u)$ for $u \in L^1$. This shows that $u$ is a fixed point of $\Phi_f$.
 Moreover, one has $\varphi = \widetilde{\varphi} + \lim\limits_{k\to\infty}\mathcal{P}^{-1}(D - q u_k)$ in $L^\infty$ which shows that $\varphi \in W^{1,2}$. 
By Theorem \ref{V.12},
the pair $\{\varphi,u\}$ is a solution of the Kohn-Sham system with distribution function~$f$.
\end{proof}
If the Kohn-Sham system with distribution function $f$ has several solutions, then it
remains unclear to which of them a sequence of solutions of Kohn-Sham
systems with distributions functions $f_j$ converges. However, if
the exchange correlation term is absent, then the result can be
improved. 
\begin{cor} \label{c-appr}
Let the assumptions of Theorem \ref{t-conv} be satisfied. 
If the exchange correlation term $V_{xc}$ is absent
and if $\{\{\varphi_j,u_j\}\}^\infty_{j=1}$ are unique solutions of
Kohn-Sham systems with distribution function $f_j$, then 
$\left\{\{\varphi_j,u_j\}\right\}^\infty_{k=1}$  converges in
$L^\infty \times L^1$ to the unique solution $\{\varphi,u\}$ of the
Kohn-Sham system with distribution function $f$.
\end{cor}
\begin{proof}
Assume that the sequence $\{\{\varphi_j,u_j\}\}^\infty_{j=1}$ does not converge to $\{\varphi,u\}$.
 In this case there is a subsequence $\{\{\varphi_k,u_k\}\}^\infty_{k=1}$ converging in
$L^\infty \times L^1$ to an element $\{\tilde \varphi,\tilde u\} \in L^\infty
 \times L^1$ which is different from $\{\varphi,u\}$. However, by Theorem
\ref{t-conv} the pair $\{\tilde \varphi,\tilde u\}$ is a solution of the Kohn-Sham
system with distribution function $f$. Since this Kohn-Sham system admits
only one solution the solution $\{\tilde \varphi,\tilde u\}$ coincides with
$\{\varphi,u\}$ -- what is a contradiction.
\end{proof}
\begin{lem}  \label{l-monotonT}
The function $f_\beta(x)=\frac {1}{\beta} \ln(1+e^{-\beta{x}})$, $x
\in \mathbb{R}$, strictly decreases in $\beta \in ]0,\infty[$.
\end{lem}
\begin{proof}
One calculates
\[
\frac {\partial}{\partial \beta}f_\beta (x) =-\frac {1}{\beta^2} \ln(1+e^{-\beta{x}}) 
-\frac {x}{\beta}\frac {e^{-\beta x}}{1+e^{-\beta x}},
\] 
what immediately shows the assertion for $x \ge 0$. Putting $-\beta x=:\gamma$, the assertion
for negative $x$ is equivalent to 
\begin{displaymath}
\frac {\gamma e^\gamma}{1+e^\gamma} < \ln \bigl(1+e^\gamma \bigr )\, ,
\quad \gamma >0, 
\end{displaymath}
which follows from $\frac {\gamma e^\gamma}{1+e^\gamma} < \gamma =\ln \bigl (e^\gamma \bigr ) <
\ln \bigl (1+e^\gamma \bigr )$.
\end{proof}

In order to apply this to the Kohn-Sham system at zero temperature, we show in the following
that the corresponding distribution function satisfies the
condition (\ref{e-convweight}).
\begin{lem} \label{l-convdistr}
Let $\{T_j\}^\infty_{j=1}$ be any positive 
sequence converging to zero. We set 
$f_j(x) := \frac{1}{\beta_j}\ln(1+e^{-\beta_j x})$ where $\beta_j = \frac{1}{kT_j}$.
Further, we set
\begin{equation}
f(x) := \left\{
\begin{array}{r@{\quad:\quad}l}
-x & x \le 0\\
0 & x > 0.
\end{array} \right.
\end{equation}
Then condition (\ref{e-convweight}) is satisfied.
\end{lem}
\begin{proof}
We have
\begin{eqnarray}
\lefteqn{
\lim_{j \mapsto \infty}\sup_{x \in [a,\infty[}|f_j(x)-f(x)|\theta(x)
\le }\nonumber\\
& &
\lim_{j \mapsto \infty}\sup_{a \le x \le 1}|f_j(x)-f(x)| +
\lim_{j \mapsto \infty}\sup_{x \ge 1}|f_j(x)-f(x)|x
\label{7.7}
\end{eqnarray}
for $a \le -1$. Obviously we have
\begin{displaymath}
\frac {1}{\beta}\ln(1+e^{-\beta x})x \le 
\frac {e^{-\beta x}x}{\beta}, 
\qquad x \ge 1.
\end{displaymath}
This shows that the second term of (\ref{7.7}) tends to zero as $j \to
\infty$. Further, it is almost obvious that $f_j(x)$ converges pointwise
to the continuous function $f$ for $j \mapsto \infty$. Because the family
 $\{f_{\beta}\}_\beta$ is monotonously decreasing in $\beta$ by the preceding lemma,
 the convergence is uniform on bounded intervals by Dini's
 theorem. This proves that the first term of (\ref{7.7}) tends to zero
 as $j \to \infty$.
\end{proof}

\section*{Acknowledgment}
H.C., H.N. and J.R. acknowledge support from the Danish 
F.N.U. grant {\it  Mathematical Physics and Partial Differential Equations}. 
This work was initiated during a visit of K.H, H.N. and J.R. at the
Department of Mathematical Sciences of the Aalborg University and they are 
thankful for kind hospitality extended to them during the work on this
paper.

\end{document}